\documentclass[letterpaper, 10 pt, conference]{ieeeconf}  
\usepackage{gen_settings}

\IEEEoverridecommandlockouts                              

\title{\LARGE \bf
A Scenario Approach to Risk-Aware Safety-Critical System Verification
}

\author{Prithvi Akella, Mohamadreza Ahmadi, and Aaron D. Ames$^{1}$
\thanks{*This work was supported by AFOSR}
\thanks{$^{1}$All authors are with the California Institute of Technology
        {\tt\small pakella@caltech.edu}}%
}

\begin{document}

\maketitle
\thispagestyle{empty}
\pagestyle{empty}

\begin{abstract}

With the growing interest in deploying robots in unstructured and uncertain environments, there has been increasing interest in factoring risk into safety-critical control development.  Similarly, the authors believe risk should also be accounted in the verification of these controllers.  In pursuit of sample-efficient methods for uncertain black-box verification then, we first detail a method to estimate the Value-at-Risk of arbitrary scalar random variables without requiring \textit{apriori} knowledge of its distribution.  Then, we reformulate the uncertain verification problem as a Value-at-Risk estimation problem making use of our prior results.  In doing so, we provide fundamental sampling requirements to bound with high confidence the volume of states and parameters for a black-box system that could potentially yield unsafe phenomena.  We also show that this procedure works independent of system complexity through simulated examples of the Robotarium.

\end{abstract}

\section{INTRODUCTION}
Safety is of critical importance to autonomous systems, and so it is natural to ask: do controllers for these systems guarantee safety in practice?  While there are multiple approaches to this verification question, our work aims to extend the optimization and sample-based approaches of~\cite{dreossi2019verifai,haesaert2017data,ghosh2018verifying,annpureddy2011s,donze2010breach}.  For most works in this vein, satisfactory robot behavior is categorized via positive evaluation of a robustness measure over the robot's state trajectory.  As a result, it is natural to phrase controller verification as the minimization of this measure over a set of interest~\cite{corso2021survey}.  Then, there exist multiple methods to solve such an optimization problem, \textit{e.g.} Bayesian Optimization~\cite{dreossi2019verifai}, Monte Carlo Sampling~\cite{annpureddy2011s}, \textit{etc}.

However, with the growing interest in deploying robots in unstructured and uncertain environments~\cite{ahmadi2020risk,mcghan2016resilient,dixit2020risk,prob_robotics}, we should likewise consider uncertainty in the robot's evolution as part of the verification procedure.  While stochastic verification has received interest in the recent past,\cite{steinhardt2012finite,santoyo2021barrier,wisniewski2017stochastic,prajna2007framework}, these works phrase the identification of a stochastic barrier certificate or probabilistic verification statement as the outcome of an optimization problem over a multi-dimensional set of interest \textit{e.g.} a coefficient space for SOS-optimization methods for barrier construction, or a parameter space for direct verification methods.  As such, these methods are still prone to dimensional scaling issues or do not offer minimal sampling or iteration requirements for the generation of their probabilistic verification statement.

\newidea{Our Contribution:} We hope to provide a step towards sample-efficient risk-aware safety-critical verification methods in a black-box verification setting.  Specifically, we construct a method that provides a fundamental sampling requirement for probabilistic verification statement generation independent of the dimension of the verification problem.  We achieve this for a general class of nonlinear systems whose dynamics are corrupted by unstructured noise.  To achieve this, we first use a scenario approach to construct a method that is guaranteed to find upper bounds to the Value-at-Risk for arbitrary scalar random variables.  Then, we re-frame the verification problem as a Value-at-Risk determination problem for a scalar random variable.  Our verification result then stems via application of our prior, Value-at-Risk estimation results.  We also prove that through our procedure, we can bound with high confidence the volume of potentially unsafe system states and parameters.

\begin{figure}[t]
    \centering
    \includegraphics[width = 0.48\textwidth]{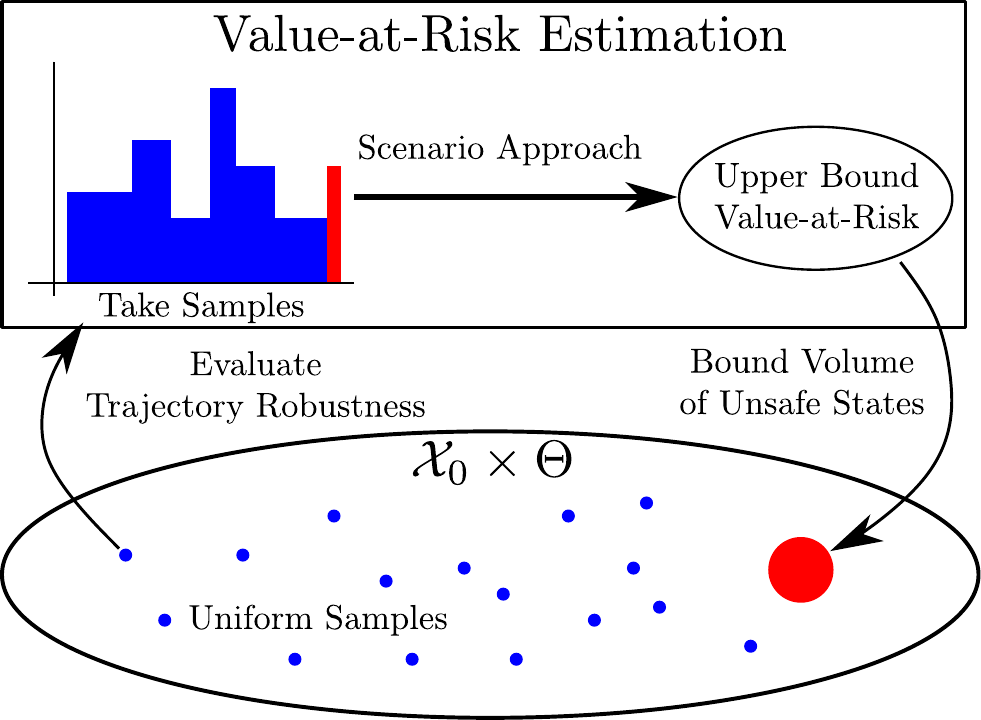} \vspace{-0.2 in}
    \caption{A flowchart of our procedure.  First, we propose a method to estimate the Value-at-Risk of a scalar-valued random variable without knowledge of its distribution.  Then, we reformulate the verification problem as a Value-at-Risk determination problem by randomly sampling system trajectories and evaluating their robustness.  In doing so, we can bound the volume of potentially problematic states and parameters with high confidence.} \vspace{-0.2 in}
    \label{fig:title}
\end{figure}

\newidea{Structural Overview:}  We will split our paper into two parts.  In the first part, Section~\ref{sec:scenario2risk}, we outline our scenario approach to Value-at-Risk estimation of scalar-valued random variables.  To facilitate doing so, Section~\ref{sec:prelims} provides some background information on both topics and formally states the estimation problem that we prove we can solve in Section~\ref{sec:high_conf_est}.  Then, Section~\ref{sec:numeric_ex} provides numerical examples of this estimation approach.  In the second part, Section~\ref{sec:risky_verification}, we recast the verification problem as a Value-at-Risk estimation problem.  Then in Section~\ref{sec:verification_properties}, we provide sample complexity results on utilizing our prior approach to solve this version of the verification problem.  Finally, in Section~\ref{sec:examples}, we show that our procedure works in simulation independent of system complexity.

\section{A Scenario Approach to Value-at-Risk}
\label{sec:scenario2risk}
The first part of our paper will detail our method to determine the Value-at-Risk (VaR) for scalar valued random variables with arbitrary distributions.  Specifically, consider a probability space $(\Omega, \mathcal{F}, P)$ with $\Omega$ the sample space, $\mathcal{F} = 2^{\Omega}$ the event space, and $P$ the probability measure.  We will devise a method to upper bound the VaR of the Random Variable (R.V.) $X: \Omega \to \mathbb{R}$ with an unknown distribution function $\pi:\mathbb{R} \to [0,1]$, \textit{i.e.} with samples $x$ of the R.V. $X$,
\begin{equation}
    \label{eq:distribution_function}
    P[X \in A \in \mathcal{F}] = \prob_{\pi}[x \in A \subset \mathbb{R}] = \int_A~\pi(s)~ds.
\end{equation}
To preface this method, we will provide a brief definition of Value-at-Risk and our scenario approach in the following section.  Here, "scenario approach" references the work done in~\cite{campi2008exact} and not the scenario approaches commonly utilized for VaR determination in the financial literature~\cite{andersson2001credit,larsen2002algorithms}.

\subsection{Preliminaries and Problem Formulation}
\label{sec:prelims}
\newidea{Value-at-Risk:} Value-at-Risk (VaR) is a statistic used quite frequently in the financial literature to determine an individual's exposure to risk~\cite{linsmeier2000value,rockafellar2000optimization}.  Both VaR and other risk measures have also been used quite frequently in the controls literature to determine safe actions in a worst-case sense~\cite{ dixit2020risk, majumdar2020should, ahmadi2022risk}.  Succinctly though, VaR is defined as follows.
\begin{definition}
\label{def:var}
Let $X$ be a scalar valued random variable with distribution function $\pi$ as per equation~\eqref{eq:distribution_function}.  The \emph{Value-at-Risk} level $\epsilon \in [0,1]$ defined $\var_{\epsilon}(X)$ is the infimum over $\zeta \in \mathbb{R}$ of $\zeta$ such that samples $x$ of $X$ lie below $\zeta$ with probability greater than or equal to $1-\epsilon$, \textit{i.e.}
\begin{equation}
    \var_{\epsilon}(X) = \inf \{\zeta~|~\prob_{\pi}[x \leq \zeta] \geq 1-\epsilon]\}.
\end{equation}
\end{definition}
\noindent Though there are other, more common risk measures in the controls literature - notably, Conditional-Value-at-Risk and Entropic-Value-at-Risk - VaR will still be useful for verification purposes as will be mentioned in a later section.

\newidea{Scenario Optimization:}  The brief description of scenario optimization in this section will stem primarily from \cite{campi2008exact}.   Scenario optimization identifies robust solutions to uncertain convex optimization problems of the following form:
\begin{equation}
    \label{eq:uncertain_program}
    \tag{UP}
    \begin{aligned}
        z^* & = \argmin_{z \in \mathbb{Z} \subset \mathbb{R}^d}~ & &c^Tz, \\
        &~~\mathrm{subject~to}~ & &z \in \mathbb{Z}_{\delta},~\delta \in \Delta.
    \end{aligned}
\end{equation}
Here, \eqref{eq:uncertain_program} is an uncertain program as $\delta \in \Delta$ is a random variable with distribution $\pi_\delta$.  Convexity is assured via assumed convexity in $\mathbb{Z}$ and $\mathbb{Z}_{\delta}$, and $\Delta$ is typically a set of infinite cardinality.  Hence, direct identification of a robust solution $z^*$ such that $z^* \in \mathbb{Z}_{\delta},~\forall~\delta \in \Delta$ is usually infeasible.

\begin{figure}[t]
    \centering
    \includegraphics[width = 0.47 \textwidth]{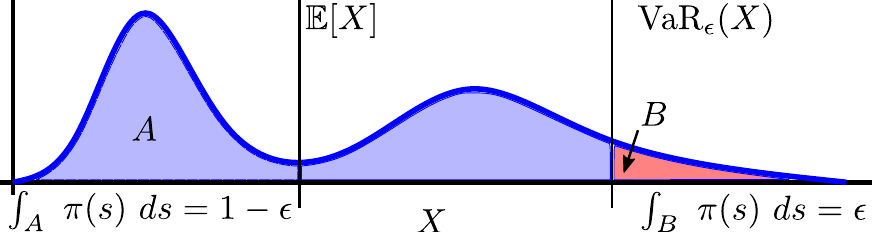}
    \caption{An example of Value-at-Risk.  For a scalar random variable $X$ with distribution $\pi$, $\var_{\epsilon}(X)$ is the point on the real line for which at least $1-\epsilon$ of the probability mass of $X$ lies to its left.} \vspace{-0.2 in}
    \label{fig:Var_example}
\end{figure}

To resolve this issue, scenario optimization solves a related optimization problem formed from an $N$-sized sample $\deltaset$ of the constraints $\delta$ and provides a probabilistic guarantee on the robustness of the corresponding solution $z^*_N$.  Specifically, given the sample scenario set $\deltaset$, we construct the following scenario program and assume that it has a solution for any $N$-sample set $\deltaset$~\cite{campi2008exact}:
\begin{equation}
    \label{eq:scenario_program}   
    \tag{RP-N}
    \begin{aligned}
        z^*_N & = \argmin_{z \in \mathbb{Z} \subset \mathbb{R}^d}~ & &c^Tz, \\
        &~~\mathrm{subject~to}~ & &z \in \mathbb{Z}_{\delta_i},~\forall~\delta_i \in \{\delta_i\}_{i=1}^N.
    \end{aligned}
\end{equation}
\begin{assumption}
\label{assump:RPN_solvability}
The scenario program~\eqref{eq:scenario_program} is solvable for any $N$-sample set $\deltaset$ and has a unique solution $z^*_N$.
\end{assumption}

However, as $z^*_N$ is the solution to~\eqref{eq:scenario_program}, there must exist a probability of sampling an uncertain parameter $\delta$ such that $z^*_N$ is not in the corresponding constraint set $\mathbb{Z}_{\delta}$.  Called the \emph{violation probability}, its definition is as follows.
\begin{definition}
\label{def:violation}
The \textit{violation probability} $V(z)$ of a given $z \in \mathbb{Z}$ is defined as the probability of sampling a constraint $\delta$ to which $z$ is not robust, \textit{i.e.} $V(z) = \prob_{\pi_\delta}[\delta~|~z \not \in \mathbb{Z}_{\delta}]$ .
\end{definition}

\noindent Then, the main result in scenario optimization upper-bounds the violation probability with high confidence.
\begin{theorem}[Adapted from Theorem 1 in~\cite{campi2008exact}]
\label{thm:scenario_opt}
Let~\eqref{eq:scenario_program} be the scenario program for~\eqref{eq:uncertain_program} formed from the $N$-sample set $\{\delta_i\}_{i=1}^N$ of the uncertain parameter $\delta$ with distribution $\pi_{\delta}$, and let $z^*_N \in \mathbb{R}^d$ be the solution to this scenario program~\eqref{eq:scenario_program}.  If Assumption~\ref{assump:RPN_solvability} holds, then $\forall~\epsilon\in[0,1]$,
\begin{equation}
    \prob^N_{\pi_\delta}[V(z^*_N) > \epsilon] \leq \sum_{i=0}^{d-1} \binom{N}{i} \epsilon^i(1-\epsilon)^{N-i}.
\end{equation}
\end{theorem}
\noindent Effectively, Theorem~\ref{thm:scenario_opt} bounds the probability that our scenario solution's violation probability $V(z^*_N)$ exceeds a cutoff $\epsilon \in [0,1]$ based on the number of samples taken $N$, the dimension $d$ of our solution $z^*_N$, and the cutoff value $\epsilon$.  This concludes our brief overview of important topics - we will now formally mention the first problem of interest.

\newidea{Formal Problem Statement:}  We wish to construct a method to determine an upper bound to the VaR of a scalar-valued R.V. $X$ whose distribution $\pi$ is unknown.  In other words:
\begin{problem}
For the R.V. $X$ with unknown distribution function $\pi$ as per~\eqref{eq:distribution_function}, devise a method to determine an upper bound $\zeta \in \mathbb{R}$ to the Value-at-Risk level $\epsilon$ for $X$ with high probability, \textit{i.e.} for some $\beta,\epsilon \in [0,1]$, find $\zeta \in \mathbb{R}$ such that,
\begin{equation}
    \prob_{\pi}[\zeta \geq \var_{\epsilon}(X)] \geq \beta.
\end{equation}
\end{problem}

\subsection{High Confidence Estimation of Value-at-Risk}
\label{sec:high_conf_est}
To determine an upper bound to the Value-at-Risk of a scalar R.V. $X$, we will reformulate this upper bound identification as an uncertain convex optimization problem with respect to samples $x$ of the R.V. $X$:
\begin{equation}
    \label{eq:upper_bound_uncertain}
    \tag{UB-UP}
    \begin{aligned}
        \zeta^* & = \argmin_{\zeta \in \mathbb{R}}~ & &\zeta, \\
        &~~\mathrm{subject~to}~ & & \zeta \geq x,~\forall~x \in \mathbb{R},~x \sim \pi.
    \end{aligned}
\end{equation}
For~\eqref{eq:upper_bound_uncertain} then, if we take an $N$-sample set of the constraints $\{x_i\}_{i=1}^N$ we can construct a scenario program.
\begin{equation}
    \label{eq:upper_bound_scenario}
    \tag{UB-RP-N}
    \begin{aligned}
        \zeta^*_N & = \argmin_{\zeta \in \mathbb{R}}~ & &\zeta, \\
        &~~\mathrm{subject~to}~ & & \zeta \geq x_i,~\forall~x_i \in \{x_i\}_{i=1}^N.
    \end{aligned}
\end{equation}

Then, our results are twofold.  First, in Lemma~\ref{lem:solvability}, we guarantee that the scenario program~\eqref{eq:upper_bound_scenario} is solvable for any $N$-sample set $\{x_i\}_{i=1}^N$ and that the solution $\zeta^*_N$ is unique.  Then, we use Lemma~\ref{lem:solvability} to prove that a solution to our scenario program~\eqref{eq:upper_bound_scenario} is an upper bound to $\var_{\epsilon}(X)$ for any $\epsilon \in [0,1]$ with high probability.
\begin{lemma}
\label{lem:solvability}
The scenario program~\eqref{eq:scenario_program} has a unique solution $\zeta^*_N$ for any $N$-sample set $\{x_i\}_{i=1}^N$.
\end{lemma}

\begin{proof}
For any $N$-sample set $\{x_i\}_{i=1}^N$, the scenario program~\eqref{eq:scenario_program} is a linear program minimizing a scalar devision variable subject to a finite number of lower bounds $b_i \in \mathbb{R}$.  This guarantees a unique solution.
\end{proof}

\begin{theorem}
\label{thm:bounding_var}
Let $\zeta^*_N$ be the solution to~\eqref{eq:upper_bound_scenario} for a set of $N$ samples $\{x_i\}_{i=1}^N$ of a R.V. $X$ with unknown distribution function $\pi$ as per~\eqref{eq:distribution_function}.  The following statement is true $\forall~\epsilon \in[0,1]$ and with $\var_{\epsilon}(X)$ as defined in Definition~\ref{def:var}:
\begin{equation}
    \prob^N_{\pi}[\zeta^*_N \geq \var_{\epsilon}(x)] \geq 1-(1-\epsilon)^N.
\end{equation}
\end{theorem}

\begin{table}[t]
    \centering
    \caption{$\min\limits_{i=1,2,\dots,30}~\zeta^{i*}_N - \var_{\epsilon}(X)$ for the distributions listed.}
    \label{tab:scenario_numerics}
    \begin{tabular}{|c|c|c|c|}
        \hline
        $\epsilon$ & $\uniform[-1,1]$ & $\normal(0,1)$ & $\chi(2)$ \\ 
        \hline
         0.01 & $\approx 0.019$ & $\approx 0.778$ & $\approx 4.56$\\
        \hline
         0.007 & $\approx 0.012$ & $\approx 0.596$ & $\approx 4.96$ \\
        \hline
         0.003 & $\approx 0.007$ & $\approx 0.509$ & $\approx 3.25$ \\
        \hline
         0.001 & $\approx 3 \times 10^{-4}$ & $\approx 0.066$ & $\approx 0.458$ \\
        \hline
    \end{tabular} \vspace{-0.2 in}
\end{table}

\begin{proof}
Lemma~\ref{lem:solvability} proves that our scenario program~\eqref{eq:upper_bound_scenario} satisfies Assumption~\ref{assump:RPN_solvability}.  Per Theorem~\ref{thm:scenario_opt} then, we have the following inequality as $d=1$ since $\zeta \in \mathbb{R}$:
\begin{equation}
    \prob^N_{\pi}[V(\zeta^*_N) > \epsilon] \leq (1-\epsilon)^N.
\end{equation}
Then by definition of the violation probability in Definition~\ref{def:violation},
\begin{equation}
    \prob^N_{\pi}[\prob_\pi[x~|~x \leq \zeta^*_N] \geq 1-\epsilon] \geq 1- (1-\epsilon)^N.
\end{equation}
Then, by the definition of Value-at-Risk in Definition~\ref{def:var}, we can restate the argument for the outer probability as follows:
\begin{equation}
    \prob^N_{\pi}[y^*_N \geq \var_{\epsilon}(X)] \geq 1- (1-\epsilon)^N.
\end{equation}
\end{proof}

\begin{figure}[t]
    \centering
    \includegraphics[width = 0.48\textwidth]{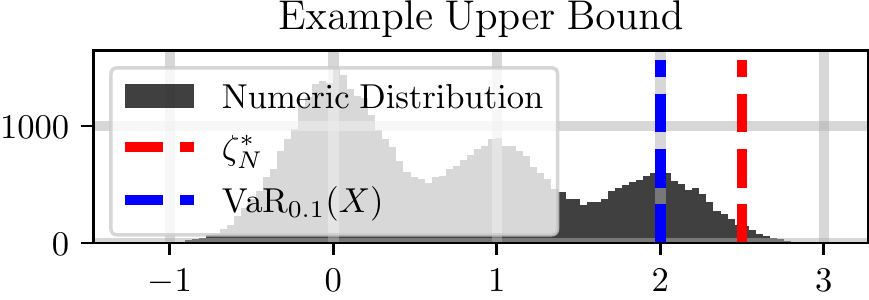} \vspace{-0.2 in}
    \caption{An example of our scenario approach to determining an upper bound to the Value-at-Risk of the R.V. whose distribution is shown.  In this case, we seek an upper bound $\zeta^*_N$ such that $\zeta^*_N \geq \var_{0.1}(X)$.} \vspace{-0.2 in}
    \label{fig:ex_upperbnd}
\end{figure}

For verification purposes then, we will randomly sample system trajectories and evaluate the robustness of those trajectories through a robustness measure $\rho$.  The positivity of this measure indicates satisfactory behavior.  Then, our goal will be to identify a lower bound $-\zeta^*_N$ such that a large fraction of the probability mass of this random robustness $R$ - $1-\epsilon$ of the probability mass - lies to the right of the calculated lower bound $-\zeta^*_N$.  This is a Value-at-Risk determination problem.  Before phrasing the verification problem in this fashion and utilizing these results, however, we will provide some examples of our scenario approach to Value-at-Risk determination and Theorem~\ref{thm:bounding_var} in the next section.

\subsection{Numerical Examples}
\label{sec:numeric_ex}
Theorem~\ref{thm:bounding_var} states that our scenario approach will identify an upper bound $\zeta^*_N$ to $\var_{\epsilon}(X)$ for any $\epsilon \in [0,1]$ and for any scalar-valued random variable $X$ whose distribution function $\pi$ is unknown.  Table~\ref{tab:scenario_numerics} shows the results of applying our procedure to identify upper bounds to the VaR of three random variables $X$ whose distributions are listed in the table.  Specifically, for each of the $30$ trials performed, we took $N=5000$ samples from the distribution listed to calculate the upper bound $\zeta^{*}_N$ to the Value-at-Risk for the R.V. X with the distribution shown.  The risk levels $\epsilon$ are shown in the table as well.  We estimated the true $\var_{\epsilon}(X)$ by taking $50000$ samples of the same random variable and reporting the value for which $1-\epsilon$ of the sampled probability mass lay to its left.  In the table, we report the difference between the minimum upper bound $\zeta^{i*}_N$ generated over $30$ trials and this true Value-at-Risk estimate.  Per Theorem~\ref{thm:bounding_var}, we expect the reported upper bound $\zeta^*_N$ to be greater than or equal to $\var_{\epsilon}(X)$ with minimum probability $0.993$ over all trials.  This is corroborated by the fact that for all $360$ upper bounding trials performed, the upper bound $\zeta^*_N \geq \var_{\epsilon}(X)$ for the corresponding risk-level $\epsilon$, as all reported differences in Table~\ref{tab:scenario_numerics} are positive.  Figure~\ref{fig:ex_upperbnd} portrays the results of one trial of our approach in determining an upper bound to the Value-at-Risk of a multi-modal random variable whose distribution is also shown.  As can be seen, even with $N=50$ samples, we can determine a high probability upper bound to the Value-at-Risk of this R.V. as well.

\section{Applications to Verification}
\label{sec:app2ver}
An interesting application of Theorem~\ref{thm:bounding_var} arises in safety-critical system verification where system evolution is partially stochastic due to un-modeled dynamics, noise, \textit{etc}.  Study in this vein is meaningful insofar as true system trajectories exhibit a nontrivial amount of stochasticity which both the authors and others in the controls community believe should be accounted for during both control development and verification~\cite{aastrom2012introduction,majumdar2020should,ahmadi2020risk,ahmadi2021constrained,santoyo2021barrier}.  To formalize study in this vein then, we will state both the specific problem and how our prior results are directly applicable in this scenario.
\subsection{Recasting Verification as a Value-at-Risk Problem}
\label{sec:risky_verification}
We will consider the general class of systems modelable by a nonlinear control system with $\mathcal{X}$ the state space, $\mathcal{U}$ the input space, and $\Theta$ a known space of parameters $\theta$ influencing the system's controller $U$.  Furthermore, we assume the system is subject to stochastic noise $\xi$ distributed by the unknown distribution $\pi_{\xi}(x,u,t)$ over $\mathbb{R}^n$, which lets us account for any deterministic or stochastic uncertainty that varies by the system state, time, or input.
\begin{align}
    \dot x & = f(x,u) + \xi, & & x \in \mathcal{X} \subset \mathbb{R}^n,~u \in \mathcal{U} \subset \mathbb{R}^m,~\label{eq:evolution}\\
    u & = U(x,\theta), & & \theta \in \Theta \subset \mathbb{R}^p, \\
    \xi & \sim \pi_{\xi}(x,u,t), & & \int_{\mathcal{X}}\pi_{\xi}(x,u,t,s)~ds = 1.
\end{align}
We will define $x^\theta_t$ as our closed-loop system solution at time $t$ - note that the parameter $\theta$ does not change over a trajectory:
\begin{equation}
    \label{eq:CL_sys}
    \dot x^\theta_t = f\left(x^\theta_t, U\left(x^\theta_t, \theta\right)\right) + \xi.
\end{equation}
Finally, $x^\theta$ corresponds to our state signal, \textit{i.e.} $x^\theta \in \signalspace$.

Then, verification work typically assumes the existence of a robustness metric $\rho$ - a function that maps state trajectory signals to the real line, with positive evaluations of the metric indicating system objective satisfaction.
\begin{definition}
\label{def:robustness}
A \textit{robustness metric} $\rho$ is a function $\rho: \signalspace \to \mathbb{R}$ such that $\rho(s) \geq 0$ only for those signals $s$ that exhibit satisfactory behavior.
\end{definition}
\noindent Examples of robustness metrics $\rho$ include the minimum value of a control barrier function $h$ over some pre-specified time horizon~\cite{ames2016control, santoyo2021barrier, srinivasan2018control}, or the robustness metrics of Signal Temporal Logic~\cite{donze2010robust, deshmukh2017robust, lindemann2019control}.  As the existence and construction of these functions has been well-studied, we will simply assume their existence for the time being.

Then, our recasting of verification as a Value-at-Risk problem stems directly from the existence of this robustness measure $\rho$.  Specifically, if we uniformly randomly sample the system's initial condition $x_0 \in \mathcal{X}_0 \subseteq \mathcal{X}$ and control parameter $\theta \in \Theta$, then the robustness of the resulting closed-loop trajectory $\rho(x^\theta)$ is a sample of some scalar-valued random variable $R$ whose distribution $\pi_R$ is unknown.  This is the same random variable setting as we had for the Value-at-Risk determination problem.  As such, we can use our prior approach to determine an upper bound in this case on $\var_{\epsilon}(-R)$ which translates to a probabilistic lower bound for the minimum robustness achievable by trajectories of this system.  This is formalized through the following definitions and problem statement.

\begin{definition}
\label{def:traj_robustness}
$R(x_0,\theta)$ is a scalar-valued random variable with distribution $\pi_{R}(x_0,\theta)$ as per~\eqref{eq:distribution_function} corresponding to the closed-loop robustness $\rho(x^\theta)$ of trajectories $x^\theta$ emanating from the initial condition $(x_0,\theta) \in \mathcal{X}_0 \times \Theta \triangleq \Phi$.
\end{definition}

\begin{definition}
\label{def:randomized_robustness}
$R$ is a scalar-valued random variable with distribution $\pi_R$ as per~\eqref{eq:distribution_function} denoting the closed-loop robustness $\rho(x^\theta)$ of trajectories $x^\theta$ whose initial condition and parameter $(x_0,\theta)$ were sampled uniformly from their combined spaces, \textit{i.e.} $(x_0,\theta) \sim \uniform[\Phi]$ where $\Phi \triangleq \mathcal{X}_0 \times \Theta$.
\end{definition}



\begin{problem}
\label{prob:2}
For the R.V. $R$ with distribution $\pi_R$ as per Definition~\ref{def:randomized_robustness} devise a method to determine a lower bound $r^* \in \mathbb{R}$ such that samples $r$ of $R$ are greater than or equal to $r^*$ with minimum probability $1-\epsilon$ for some $\epsilon \in [0,1]$, \textit{i.e.}
\begin{equation}
    \label{eq:ineq_rstar}
    \prob_{\pi_R}[r~|~r \geq r^*] \geq 1-\epsilon.
\end{equation}
\end{problem}

As identification of such a lower bound $r^*$ is non-standard in existing verification literature, we will first mention the utility in the identification of such a probabilistic lower bound $r^*$.  Then, we will mention how we can use our prior results to determine such a lower bound and provide a minimum sample requirement for doing so.

\begin{figure}[t]
    \centering
    \hspace{-0.1 mm}\includegraphics[width = 0.48\textwidth]{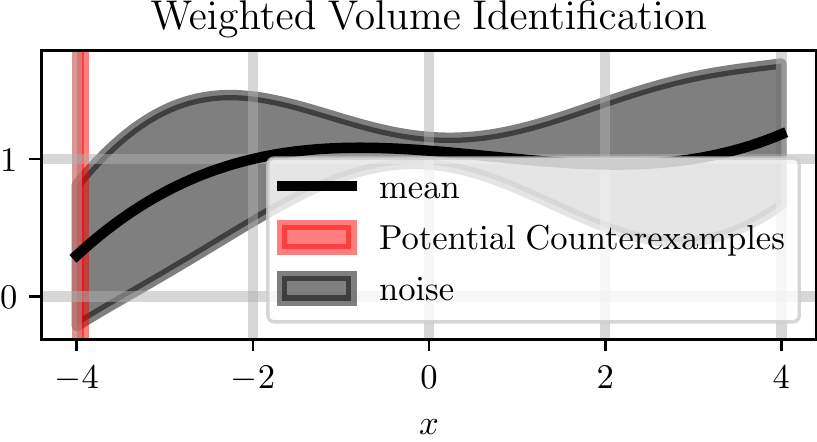}
    \caption{The above figure provides context for why we choose to identify a probabilistic lower bound $r^*$ satisfying the inequality in Problem~\ref{prob:2} for some $\epsilon \in [0,1]$.  Doing so let's us upper bound by $\epsilon$ the weighted volume of the states in the red region shown.  For verification purposes, this bounds the weighted volume of the initial condition and parameters $(x_0,\theta)$ that could yield trajectories whose robustness $r < r^*$ as stated in Corollary~\ref{cor:weight_volume}} \vspace{-0.2 in}
    \label{fig:motivator}
\end{figure}

\subsection{Properties of Randomized Verification}
\label{sec:verification_properties}

\newidea{Why Identify $\mathbf{r^*}$?} We propose to identify a probabilistic lower bound $r^*$ for the random variable $R$ as it will let us bound the weighted volume of the initial conditions and parameters $(x_0,\theta) \in \Phi$ which have the potential of yielding trajectories whose robustness $r < r^*$.  A motivating example of this is shown in Figure~\ref{fig:motivator}, and we can show this volume bounding as follows.  First, we define a function $B: \mathbb{R} \to \mathbb{R}$ outputting the total probability of sampling a trajectory whose robustness $r < y$ a scalar cutoff in $\mathbb{R}$.  In what follows, $\pi_R(x_0,\theta)$ is the distribution of the random variable $R(x_0,\theta)$ as per Definition~\ref{def:traj_robustness}:
\begin{align}
    B(y) & = \int_{\Phi}\frac{1}{\beta} \int_{-\infty}^{y}~\pi_{R}(x_0,\theta,s)~ds~d(x_0,\theta), \label{eq:risky_sum}\\
    \beta & = \int_{\Phi}~1~d(x_0,\theta).
\end{align}
Then, the formal statement of $r^*$ being a holistic characterization of system behavior will follow.
\begin{proposition}
\label{prop:risk_bounding}
Let $r^*$ satisfy~\eqref{eq:ineq_rstar} for some $\epsilon \in [0,1]$ and let $B$ be as per~\eqref{eq:risky_sum}.  If $r^* \geq 0$, then $B(0) \leq \epsilon$.
\end{proposition}
\begin{proof}
As $r^*$ satisfies the inequality in Problem~\ref{prob:2} for some $\epsilon \in [0,1]$ we have the following integral inequality:
\begin{equation}
    \int_{\Omega}\frac{1}{\beta} \int_{r^*}^{\infty}~\pi_{R}(x_0,\theta,s)~ds~d(x_0,\theta) \geq 1-\epsilon.
\end{equation}
Furthermore, as $\pi_R$ is a valid probability distribution,
\begin{equation}
     \int_{\Omega}\frac{1}{\beta} \int_{r^*}^{\infty}~\pi_{R}(x_0,\theta,s)~ds~d(x_0,\theta) + B(r^*) = 1.
\end{equation}
As a result, $B(r^*) \leq \epsilon$.  Then, as $\pi_R(x_0,\theta)$ is a valid probability distribution, \textit{i.e.} $\pi_{R}(x_0,\theta,s) \in [0,1],~\forall~s \in \mathbb{R}$, we also know that if $a \leq b$, then $B(a) \leq B(b)$.  As a result, since $r^* \geq 0$ and $B(r^*) \leq \epsilon$, so to is $B(0) \leq \epsilon$.
\end{proof}

As a result, if we fix a minimum probability by choosing an $\epsilon \in [0,1]$ and find that the corresponding probabilistic robustness lower bound $r^* \geq 0$, then the total probability $B(0)$ of sampling a trajectory whose robustness $r < 0$ is bounded above by $\epsilon$.  As we uniformly sample initial conditions and parameters $(x_0,\theta)$, this total probability $B(0)$ also corresponds to the weighted volume of initial conditions and parameters $(x_0,\theta) \in \Phi$ which could yield trajectories $x^\theta$ whose robustness $\rho(x^\theta) < 0$.  The weights $w(x_0,\theta)$ for this weighted volume is the probability of that initial condition and parameter pair $(x_0,\theta)$ realizing such a trajectory, \textit{i.e.}
\begin{align}
    B(0) & = \int_{\Phi}~\frac{w(x_0,\theta)}{\beta}~d(x_0,\theta), \\
    w(x_0,\theta) & = \int_{-\infty}^0~\pi_{R}(x_0,\theta,s)~ds.
\end{align}

If all uncertainty $\xi$ in our system evolution~\eqref{eq:evolution} is deterministic, then $B(0)$ directly corresponds to the volume fraction of those initial condition and parameter pairs that yield trajectories whose robustness $\rho(x^\theta) < 0$.  To formalize this, we will define the set of the problematic initial condition and parameter pairs $F$ and a function $\mathcal{V}$ identifying the volume fraction of a given subset $A$ of $\Phi$:
\begin{align}
    \label{eq:failure_set}
    F(y) & = \left\{(x_0,\theta) \in \Phi~\big|~\prob_{\pi_R(x_0,\theta)}[r~|~r \leq y] \neq 0 \right\}, \\
    \mathcal{V}(A) & = \frac{\int_A~1~d(x_0,\theta)}{\int_{\Omega}~1~d(x_0,\theta)}. \label{eq:volume_fraction}
\end{align}
We will also require definition of the dirac-delta function $\delta$:
\begin{equation}
    \label{eq:dirac_delta}
    \delta:\mathbb{R} \to \mathbb{R} \suchthat \int_{a}^b \delta(x)~dx = \begin{cases}
    1 & \mbox{if~} x \in [a,b], \\
    0 & \mbox{else}.
    \end{cases}
\end{equation}
Then we have the following corollary bounding the volume fraction of the problematic states $\mathcal{V}(F(0))$:
\begin{corollary}
\label{cor:weight_volume}
Let $r^*$ satisfy~\eqref{eq:ineq_rstar}, $F$ be as in equation~\eqref{eq:failure_set}, $\mathcal{V}$ be as in equation~\eqref{eq:volume_fraction}, $\epsilon \in [0,1]$, and $\pi_{R}(x_0,\theta) = \delta(s)$ with $\delta$ as per~\eqref{eq:dirac_delta} for some $s \in \mathbb{R}$ which is (perhaps) different $\forall~(x_0,\theta) \in \Phi$.  If $r^* \geq 0$, then $\mathcal{V}(F(0)) \leq \epsilon$.
\end{corollary}
\begin{proof}
First, via Proposition~\ref{prop:risk_bounding} we have via our assumptions that $B(0) \leq \epsilon$.  This results in the following inequality:
\begin{equation}
    \int_{\Omega}\frac{1}{\beta}\int_{-\infty}^{0}~\pi_{R}(x_0,\theta,s)~ds~d(x_0,\theta) \leq \epsilon.
\end{equation}
As we assume $\pi_{R}(x_0,\theta) = \delta(s)$ for some $s \in \mathbb{R}$, the above integral inequality changes to the following:
\begin{equation}
    \int_{\Omega}\frac{1}{\beta}\int_{-\infty}^{0}~\delta(s)~ds~d(x_0,\theta) \leq \epsilon.
\end{equation}
Then, the result stems by definition of the problematic set $F$ in~\eqref{eq:failure_set} and the dirac-delta function $\delta$ in~\eqref{eq:dirac_delta}.
\begin{equation}
    \frac{\int_{F(0)}~1~d(x_0,\theta)}{\beta} = \frac{\int_{F(0)}~1~d(x_0,\theta)}{\int_{\Omega}~1~d(x_0,\theta)} = \mathcal{V}(F(0))  \leq \epsilon.
\end{equation}
\end{proof}

\newidea{Simplicity in Finding $\mathbf{r^*}$:} Now that we have motivated why we might want to find such a probabilistic lower bound $r^*$, it remains to find such a lower bound.  Here we can leverage our prior results in the following corollary.
\begin{corollary}
\label{cor:verification_risk}
Let $\zeta^*_N$ be the solution to~\eqref{eq:upper_bound_scenario} for an $N$-sample set $\{x_i=-r_i\}_{i=1}^N$ of the random variable $R$ with distribution $\pi_R$ as per Definition~\ref{def:randomized_robustness}.  For all $\epsilon \in [0,1]$,
\begin{equation}
    \prob^N_{\pi_R}\left[\prob_{\pi_R}[r~|~r \geq -\zeta^*_N] \geq 1-\epsilon \right] \geq 1-(1-\epsilon)^N.
\end{equation}
\end{corollary}
\begin{proof}
By Theorem~\ref{thm:bounding_var} and replacing $\var_{\epsilon}(X)$ with its appropriate definition in this context we have the following:
\begin{equation}
    \prob^N_{\pi_{R}}\left[\prob_{\pi_{R}}\left[r~|~\zeta^*_N \geq -r \right] \right]\geq 1-(1-\epsilon)^N.
\end{equation}
The desired result stems from flipping the inequality.
\end{proof}

Corollary~\ref{cor:verification_risk} tells us that we can identify an estimate $-\zeta^*_N$ to our desired statistic $r^*$ for any confidence level $1-\epsilon$ if we take a sufficiently large number of samples $N$ of the random variable $R$.  However, it does not state how many samples $N$ are required to determine this estimate with high probability.  Theorem~\ref{thm:sample_req} formalizes this sample requirement.  Specifically, Theorem~\ref{thm:sample_req} states that the number of samples $N$ required to achieve high confidence $\gamma$ in our estimate $-\zeta^*_N$ is only a function of the desired risk level $\epsilon$ and the desired confidence $\gamma$.

\begin{theorem}
\label{thm:sample_req}
Let $\epsilon,\gamma \in [0,1]$, and let $\zeta^*_N$ be the solution to~\eqref{eq:upper_bound_scenario} for an $N$-sample set $\{x_i = -r_i\}_{i=1}^N$ of the R.V. $R$ with distribution $\pi_R$ as per Definition~\ref{def:randomized_robustness}.  If
\begin{equation}
    N \geq \frac{\log(1-\gamma)}{\log(1-\epsilon)},
\end{equation}
then,
\begin{equation}
    \label{eq:gamma_confidence}
    \prob^N_{\pi_{R}}\left[\prob_{\pi_{R}}[r~|~r \geq -\zeta^*_N] \geq 1-\epsilon \right] \geq \gamma.
\end{equation}
\end{theorem}
\begin{proof}
Via Corollary~\ref{cor:verification_risk} we have the following inequality:
\begin{equation}
    \prob^N_{\pi_R}\left[\prob_{\pi_R}[r~|~r \geq -\zeta^*_N] \geq 1-\epsilon \right] \geq 1-(1-\epsilon)^N.
\end{equation}
As $1 - \epsilon \in [0,1]$, if $N \geq \frac{\log(1-\gamma)}{\log(1-\epsilon)}$ then 
\begin{equation}
    \prob^N_{\pi_R}\left[\prob_{\pi_R}[r~|~r \geq -\zeta^*_N] \geq 1-\epsilon \right] \geq 1-(1-\epsilon)^{\frac{\log(1-\gamma)}{\log(1-\epsilon)}}.
\end{equation}
Simplifying the right-hand side of the above inequality provides the desired result.
\end{proof}

Notably, Theorem~\ref{thm:sample_req}'s result is independent of the dimension of the system's state and parameter space, \textit{i.e.} independent of $\ell$ where $\mathcal{X}_0 \times \Theta \subseteq \mathbb{R}^\ell$.  This is why we claim we have made the first step to sample-efficient risk-aware safety-critical system verification.  Independent of system complexity, Theorem~\ref{thm:sample_req} identifies the minimum number of samples required to verify system behavior for arbitrary robustness measures satisfying Definition~\ref{def:robustness}.  Both this dimensionality independence and the general results of Theorem~\ref{thm:sample_req} will be shown through simulated examples in the following section.

\begin{figure}[t]
    \centering
    \includegraphics[width = 0.48\textwidth]{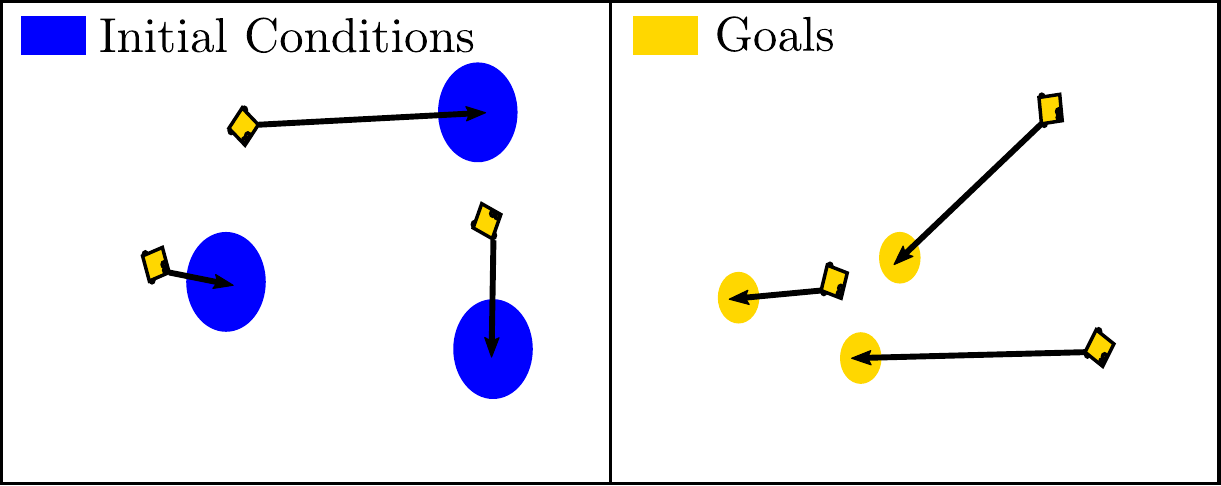}
    \caption{Setting for all robotarium simulations, shown for the three robot system.  Each data collection phase consists of the following steps: (1) Sample initial conditions uniformly over the initial condition space $\mathcal{X}_0$~\eqref{eq:init_space} and send robots there; (2) Sample goals uniformly over the parameter space $\Theta$~\eqref{eq:big_theta} and record state trajectory of robots moving to goal for $30$ seconds; (3) Calculate robustness with robustness measure $\rho$~\eqref{eq:robustness}.} \vspace{-0.2 in}
    \label{fig:my_label}
\end{figure}

\section{Simulated Examples}
\label{sec:examples}
In this section, we will provide a few simulated examples illustrating the results of Theorem~\ref{thm:sample_req} and its lack of dependence on the dimension of the state and parameter space over which a system is to be verified.  For this purpose, we will use the Robotarium~\cite{robotarium} as a case study wherein all robots can be modeled via unicycle dynamics:
\begin{gather}
    x = \begin{bmatrix}
    x,\\
    y, \\
    \theta
    \end{bmatrix},~\dot x = 
    \begin{bmatrix}
    v \cos(\theta), \\
    v \sin(\theta), \\
    \omega,
    \end{bmatrix},~u = [v, \omega]^T, \\
    \mathcal{X} = [-1,1] \times [-0.6,0.6] \times [0,2 \pi],~P = [I_2,~\mathbf{0}_{2x1}]. \label{eq:setting}
\end{gather}
Each robot has a Lyapunov-based controller that drives it to its desired orientation $x_d \in \mathcal{X}$.  When multiple robots are asked to ambulate in the same, confined space, their control inputs are filtered in a barrier-based quadratic program to ensure that the robots never collide~\cite{ames2016control}.  This barrier-based filter provides a natural robustness measure as per Definition~\ref{def:robustness}.  Specifically, if we concatenate the state vector for all $N_R$ robots $\mathbf{x}^T=[x^{1T},x^{2T},\dots,x^{N_RT}]$, we can generate a candidate barrier function that the system is to keep positive over the course of its evolution - here $P$ is defined in~\eqref{eq:setting}:
\begin{equation}
    \label{eq:avoid}
    h_g(\mathbf{x}) = \min_{i \neq j,~i,j \in [1,2,\dots,N_R]}~\|P(x^i - x^j)\| - 0.15.
\end{equation}
Furthermore, we also want the robots to reach their respective goals $g^j \in [-1,1] \times [-0.6, 0.6]$ and can quantify goal satisfaction through another barrier-like function as well.
\begin{equation}
    \label{eq:future}
    h_f(\mathbf{x}) = \max_{i \in [1,2,\dots,N_R]}~0.1 - \|Px^i - g^i\|.
\end{equation}
Then, our candidate robustness measure for a state-signal $\mathbf{x}^{\theta}$ where $\theta^T = [g^{1T},g^{2T},\dots,g^{N_RT}]$ is as follows:
\begin{equation}
    \label{eq:robustness}
    \rho(\mathbf{x}^{\theta}) = \min\left\{\min_{t \in [0,30]}~h_f\left( \mathbf{x}^\theta_t\right), \max_{t \in [0,30]}~h_g\left( \mathbf{x}^\theta_t\right) \right\}.
\end{equation}
If $\rho(\mathbf{x}^\theta) \geq 0$ then, we know that all $N_R$ robots stayed at least $0.15$ meters from each other for $30$ seconds - as $h_g(\mathbf{x}^\theta_t) \geq 0,~\forall~t \in [0,30]$ - and reached within $0.1$ meters of their goal within $30$ seconds - as $h_f(\mathbf{x}^\theta_t) \geq 0,$ for some $t \in [0,30]$.

\newidea{Verifying Theorem~\ref{thm:sample_req}:} As per Problem~\ref{prob:2}, we aim to determine a high probability lower bound $-\zeta^*_N$ to the robustness $r$ of closed-loop trajectories $\mathbf{x}^\theta$ over the initial condition and parameter spaces below for an $N_R = 3$ robot system:
\begin{align}
    \mathcal{X}_0  & = \{\mathbf{x} \in \mathcal{X}^{N_R}~|~h_g(\mathbf{x}) \geq 0.3 \}, \label{eq:init_space} \\
    \Theta & = \{P\mathbf{x} \in \mathcal{X}^{N_R}~|~h_g(\mathbf{x}) \geq 0.3 \}. \label{eq:big_theta}
\end{align}
Furthermore, we want to be at least $(1-10^{-6})\times100\%$ confident in our lower bound $-\zeta^*_N$ such that at least $97.25\%$ of the probability mass of our randomized robustness $R$ is greater than or equal to $-\zeta^*_N$.  Per Theorem~\ref{thm:sample_req} this implies $\gamma = 1-10^{-6}$ and $\epsilon = 0.0275$, and as such, we require $N \geq 496$ samples of this R.V. $R$ to calculate a lower bound $-\zeta^*_N$ that satisfies these criteria.  As a result, we uniformly sampled $N = 500$ initial condition and parameter pairs $(\mathbf{x}_0,\theta) \in \mathcal{X}_0 \times \Theta$, simulated each trajectory for $30$ seconds, recorded the trajectory's robustness $r = \rho(\mathbf{x}^\theta)$, and calculated a lower bound $-\zeta^*_N$ via~\eqref{eq:upper_bound_scenario} yielding $-\zeta^*_N = 0.0103$.

To verify that this lower bound $-\zeta^*_N$ satisfies the criteria that at least $97.25\%$ of the probability mass of $R$ is greater than or equal to $-\zeta^*_N$ we took another $20000$ randomly chosen robustness samples via the prior method and recorded the fraction $l$ of those trajectories whose robustness $r < -\zeta^*_N$.  In doing so $l = 0.0013$.  As such, we estimate that
\begin{equation}
    \prob_{\pi_R}[r~|r \geq -\zeta^*_N] \geq 1-0.0013 \geq 1-\epsilon,
\end{equation}
as desired.

However, we only expect to determine such a lower bound to a certain confidence level as well.  Specifically, we are $\gamma = 1-10^{-6}$ confident that our method will produce a lower bound $-\zeta^*_N$ that satisfies the aforementioned probabilistic inequality.  To show this, we repeated this same lower bound identification process $50$ times.  Indeed, every single time we identified a lower bound that satisfied the desired probabilistic inequality - this harmonizes with the notion that we expect to identify an unsatisfactory lower bound with probability $10^{-6}$ based on our approach.  The minimum such lower bound $-\zeta^*_N = -2.26$ and the maximum probability mass of robustness values $r < -\zeta^*_N$ was $l = 0.0114 < \epsilon = 0.0275$.

\newidea{Dimensionality Scaling:}  The prior vein of thought shows that our method is repeatable and reliable in its ability to determine lower bounds $-\zeta^*_N$ for the random variable $R$ corresponding to the robustness of a three-robot Robotarium system.  However, what if we changed the number of robots and increased the dimension of the initial condition and parameter spaces?  As per Theorem~\ref{thm:sample_req}, if we require similar confidence $\gamma$ in the identification of a lower bound for the same probabilistic cutoff $\epsilon$, the number of required samples $N$ should remain the same.  As such, we repeated the prior data-collection procedure now for an $N_R = 6$ robot system and collected the same $N=500$ samples per trial as prior.  Over all $50$ trials, the maximum calculated lower bound $-\zeta^*_N = 0.03698$, and the maximum probability mass of robustness values $r < -\zeta^*_N$ was $l = 0.01455 < \epsilon = 0.0275$.  As a result, our method not only reliably and repeatably produces probabilistic lower bounds $-\zeta^*_N$ for the random variable $R$, but it is also doesn't scale with increasing dimension of the state and parameter space $\mathcal{X}_0 \times \Theta$ as mentioned in Theorem~\ref{thm:sample_req}.

\section{Conclusion}
We present a scenario approach to risk-aware safety-critical system verification in a two-step fashion.  First, we detail a scenario method to estimate the Value-at-Risk of arbitrary scalar-valued random variables whose distribution is unknown.  Then, we reframe the verification problem as a Value-at-Risk determination problem making use of our prior estimation results.  In doing so, we bound with high probability the volume of initial condition and parameter pairs that could yield unsafe trajectories.  We also provide a minimum sampling requirement for this approach, independent of system complexity, and show that our results hold for simulated systems.  As future work, we hope to factor in system models and existing control techniques to minimize the sample requirement for our procedure.
\bibliographystyle{IEEEtran}
\bibliography{IEEEabrv,bib_works}

\begin{thebibliography}{10}
\providecommand{\url}[1]{#1}
\csname url@rmstyle\endcsname
\providecommand{\newblock}{\relax}
\providecommand{\bibinfo}[2]{#2}
\providecommand\BIBentrySTDinterwordspacing{\spaceskip=0pt\relax}
\providecommand\BIBentryALTinterwordstretchfactor{4}
\providecommand\BIBentryALTinterwordspacing{\spaceskip=\fontdimen2\font plus
\BIBentryALTinterwordstretchfactor\fontdimen3\font minus
  \fontdimen4\font\relax}
\providecommand\BIBforeignlanguage[2]{{%
\expandafter\ifx\csname l@#1\endcsname\relax
\typeout{** WARNING: IEEEtran.bst: No hyphenation pattern has been}%
\typeout{** loaded for the language `#1'. Using the pattern for}%
\typeout{** the default language instead.}%
\else
\language=\csname l@#1\endcsname
\fi
#2}}

\bibitem{dreossi2019verifai}
T.~Dreossi, D.~J. Fremont, S.~Ghosh, E.~Kim, H.~Ravanbakhsh,
  M.~Vazquez-Chanlatte, and S.~A. Seshia, ``Verifai: A toolkit for the formal
  design and analysis of artificial intelligence-based systems,'' in
  \emph{International Conference on Computer Aided Verification}.\hskip 1em
  plus 0.5em minus 0.4em\relax Springer, 2019, pp. 432--442.

\bibitem{haesaert2017data}
S.~Haesaert, P.~M. Van~den Hof, and A.~Abate, ``Data-driven and model-based
  verification via bayesian identification and reachability analysis,''
  \emph{Automatica}, vol.~79, pp. 115--126, 2017.

\bibitem{ghosh2018verifying}
S.~Ghosh, F.~Berkenkamp, G.~Ranade, S.~Qadeer, and A.~Kapoor, ``Verifying
  controllers against adversarial examples with bayesian optimization,'' in
  \emph{2018 IEEE International Conference on Robotics and Automation
  (ICRA)}.\hskip 1em plus 0.5em minus 0.4em\relax IEEE, 2018, pp. 7306--7313.

\bibitem{annpureddy2011s}
Y.~Annpureddy, C.~Liu, G.~Fainekos, and S.~Sankaranarayanan, ``S-taliro: A tool
  for temporal logic falsification for hybrid systems,'' in \emph{International
  Conference on Tools and Algorithms for the Construction and Analysis of
  Systems}.\hskip 1em plus 0.5em minus 0.4em\relax Springer, 2011, pp.
  254--257.

\bibitem{donze2010breach}
A.~Donz{\'e}, ``Breach, a toolbox for verification and parameter synthesis of
  hybrid systems,'' in \emph{International Conference on Computer Aided
  Verification}.\hskip 1em plus 0.5em minus 0.4em\relax Springer, 2010, pp.
  167--170.

\bibitem{corso2021survey}
A.~Corso, R.~Moss, M.~Koren, R.~Lee, and M.~Kochenderfer, ``A survey of
  algorithms for black-box safety validation of cyber-physical systems,''
  \emph{Journal of Artificial Intelligence Research}, vol.~72, pp. 377--428,
  2021.

\bibitem{ahmadi2020risk}
M.~Ahmadi, M.~Ono, M.~D. Ingham, R.~M. Murray, and A.~D. Ames, ``Risk-averse
  planning under uncertainty,'' in \emph{2020 American Control Conference
  (ACC)}.\hskip 1em plus 0.5em minus 0.4em\relax IEEE, 2020, pp. 3305--3312.

\bibitem{mcghan2016resilient}
C.~L. McGhan, T.~Vaquero, A.~R. Subrahmanya, O.~Arslan, R.~Murray, M.~D.
  Ingham, M.~Ono, T.~Estlin, B.~Williams, and M.~Elaasar, ``The resilient
  spacecraft executive: An architecture for risk-aware operations in uncertain
  environments,'' in \emph{Aiaa Space 2016}, 2016, p. 5541.

\bibitem{dixit2020risk}
A.~Dixit, M.~Ahmadi, and J.~W. Burdick, ``Risk-sensitive motion planning using
  entropic value-at-risk,'' \emph{arXiv preprint arXiv:2011.11211}, 2020.

\bibitem{prob_robotics}
S.~Thrun, W.~Burgard, and D.~Fox, \emph{Probabilistic Robotics (Intelligent
  Robotics and Autonomous Agents)}.\hskip 1em plus 0.5em minus 0.4em\relax The
  MIT Press, 2005.

\bibitem{steinhardt2012finite}
J.~Steinhardt and R.~Tedrake, ``Finite-time regional verification of stochastic
  non-linear systems,'' \emph{The International Journal of Robotics Research},
  vol.~31, no.~7, pp. 901--923, 2012.

\bibitem{santoyo2021barrier}
C.~Santoyo, M.~Dutreix, and S.~Coogan, ``A barrier function approach to
  finite-time stochastic system verification and control,'' \emph{Automatica},
  vol. 125, p. 109439, 2021.

\bibitem{wisniewski2017stochastic}
R.~Wisniewski and M.~L. Bujorianu, ``Stochastic safety analysis of stochastic
  hybrid systems,'' in \emph{2017 IEEE 56th Annual Conference on Decision and
  Control (CDC)}.\hskip 1em plus 0.5em minus 0.4em\relax IEEE, 2017, pp.
  2390--2395.

\bibitem{prajna2007framework}
S.~Prajna, A.~Jadbabaie, and G.~J. Pappas, ``A framework for worst-case and
  stochastic safety verification using barrier certificates,'' \emph{IEEE
  Transactions on Automatic Control}, vol.~52, no.~8, pp. 1415--1428, 2007.

\bibitem{campi2008exact}
M.~C. Campi and S.~Garatti, ``The exact feasibility of randomized solutions of
  uncertain convex programs,'' \emph{SIAM Journal on Optimization}, vol.~19,
  no.~3, pp. 1211--1230, 2008.

\bibitem{andersson2001credit}
F.~Andersson, H.~Mausser, D.~Rosen, and S.~Uryasev, ``Credit risk optimization
  with conditional value-at-risk criterion,'' \emph{Mathematical programming},
  vol.~89, no.~2, pp. 273--291, 2001.

\bibitem{larsen2002algorithms}
N.~Larsen, H.~Mausser, and S.~Uryasev, ``Algorithms for optimization of
  value-at-risk,'' in \emph{Financial engineering, E-commerce and supply
  chain}.\hskip 1em plus 0.5em minus 0.4em\relax Springer, 2002, pp. 19--46.

\bibitem{linsmeier2000value}
T.~J. Linsmeier and N.~D. Pearson, ``Value at risk,'' \emph{Financial Analysts
  Journal}, vol.~56, no.~2, pp. 47--67, 2000.

\bibitem{rockafellar2000optimization}
R.~T. Rockafellar, S.~Uryasev, \emph{et~al.}, ``Optimization of conditional
  value-at-risk,'' \emph{Journal of risk}, vol.~2, pp. 21--42, 2000.

\bibitem{majumdar2020should}
A.~Majumdar and M.~Pavone, ``How should a robot assess risk? towards an
  axiomatic theory of risk in robotics,'' in \emph{Robotics Research}.\hskip
  1em plus 0.5em minus 0.4em\relax Springer, 2020, pp. 75--84.

\bibitem{ahmadi2022risk}
M.~Ahmadi, X.~Xiong, and A.~D. Ames, ``Risk-averse control via cvar barrier
  functions: Application to bipedal robot locomotion,'' \emph{IEEE Control
  Systems Letters}, vol.~6, pp. 878--883, 2022.

\bibitem{aastrom2012introduction}
K.~J. {\AA}str{\"o}m, \emph{Introduction to stochastic control theory}.\hskip
  1em plus 0.5em minus 0.4em\relax Courier Corporation, 2012.

\bibitem{ahmadi2021constrained}
M.~Ahmadi, U.~Rosolia, M.~D. Ingham, R.~M. Murray, and A.~D. Ames,
  ``Constrained risk-averse markov decision processes,'' in \emph{The 35th AAAI
  Conference on Artificial Intelligence (AAAI-21)}, 2021.

\bibitem{ames2016control}
A.~D. Ames, X.~Xu, J.~W. Grizzle, and P.~Tabuada, ``Control barrier function
  based quadratic programs for safety critical systems,'' \emph{IEEE
  Transactions on Automatic Control}, vol.~62, no.~8, pp. 3861--3876, 2016.

\bibitem{srinivasan2018control}
M.~Srinivasan, S.~Coogan, and M.~Egerstedt, ``Control of multi-agent systems
  with finite time control barrier certificates and temporal logic,'' in
  \emph{2018 IEEE Conference on Decision and Control (CDC)}.\hskip 1em plus
  0.5em minus 0.4em\relax IEEE, 2018, pp. 1991--1996.

\bibitem{donze2010robust}
A.~Donz{\'e} and O.~Maler, ``Robust satisfaction of temporal logic over
  real-valued signals,'' in \emph{International Conference on Formal Modeling
  and Analysis of Timed Systems}.\hskip 1em plus 0.5em minus 0.4em\relax
  Springer, 2010, pp. 92--106.

\bibitem{deshmukh2017robust}
J.~V. Deshmukh, A.~Donz{\'e}, S.~Ghosh, X.~Jin, G.~Juniwal, and S.~A. Seshia,
  ``Robust online monitoring of signal temporal logic,'' \emph{Formal Methods
  in System Design}, vol.~51, no.~1, pp. 5--30, 2017.

\bibitem{lindemann2019control}
L.~Lindemann and D.~V. Dimarogonas, ``Control barrier functions for multi-agent
  systems under conflicting local signal temporal logic tasks,'' \emph{IEEE
  control systems letters}, vol.~3, no.~3, pp. 757--762, 2019.

\bibitem{robotarium}
S.~Wilson, P.~Glotfelter, L.~Wang, S.~Mayya, G.~Notomista, M.~Mote, and
  M.~Egerstedt, ``The robotarium: Globally impactful opportunities, challenges,
  and lessons learned in remote-access, distributed control of multirobot
  systems,'' \emph{IEEE Control Systems Magazine}, vol.~40, no.~1, pp. 26--44,
  2020.

\end{thebibliography}

\end{document}